\documentclass[letterpaper,11pt]{article}
\usepackage{amssymb,fullpage}
\usepackage{times,xspace}

\usepackage{xspace}

\usepackage{graphicx}
\usepackage{amsmath}
\usepackage{amsfonts}
\usepackage{amssymb}

\newcommand{\Xomit}[1]{ }

\newtheorem{theorem}{Theorem}

\newtheorem{corollary}[theorem]{Corollary}

\newtheorem{lemma}[theorem]{Lemma}

\newenvironment{proof}[1][Proof]{\textbf{#1.} }{\ \rule{0.5em}{0.5em}}
\newcommand{\eps}{\varepsilon}
\newcommand{\x}{\mathbf x}
\newcommand{\Y}{\mathbf Y}
\newcommand{\X}{\mathbf X}
\newcommand{\Z}{\mathbf Z}

\newcommand{\R}{{\cal{R}}}
\newcommand{\W}{{\cal{W}}}

\newcommand{\C}{{\cal{C}}}
\newcommand{\pr}{{\sc{BPCC}}}

\newcommand{\bpr}{{\sc{BPR}}}

\newcommand{\opt}{\mbox{\textsc{opt}}}

\newcommand{\alg}{\mbox{\textsc{alg}}}
\newcommand{\apx}{\mbox{\textsc{apx}}}

\begin{document}

\title{AFPTAS results for common variants of bin packing:\\ A  new method to handle the small items}

\author{Leah Epstein\thanks{Department of Mathematics, University of Haifa, 31905 Haifa, Israel.
{\tt lea@math.haifa.ac.il}.} \and Asaf Levin
\thanks{Chaya fellow. Faculty of Industrial Engineering and Management, The Technion,
32000 Haifa, Israel. {\tt levinas@ie.technion.ac.il.}}}
\date{}

\maketitle

\vspace{-0.7cm}
\begin{abstract}
We consider two well-known natural variants of bin packing, and
show that these packing problems admit asymptotic fully polynomial
time approximation schemes (AFPTAS). In bin packing problems, a
set of one-dimensional items of size at most 1 is to be assigned
(packed) to subsets of sum at most 1 (bins). It has been known for
a while that the most basic problem admits an AFPTAS. In this
paper, we develop methods that allow to extend this result to
other variants of bin packing. Specifically, the problems which we
study in this paper, for which we design asymptotic fully
polynomial time approximation schemes, are the following. The
first problem is {\it Bin packing with cardinality constraints},
where a parameter $k$ is given, such that a bin may contain up to
$k$ items. The goal is to minimize the number of bins used. The
second problem is {\it Bin packing with rejection}, where every
item has a rejection penalty associated with it. An item needs to
be either packed to a bin or rejected, and the goal is to minimize
the number of used bins plus the total rejection penalty of
unpacked items. This resolves the complexity of two important
variants of the bin packing problem. Our approximation schemes use
a novel method for packing the small items. Specifically, we
introduce the new notion of {\it windows}. A window is a space in
which small items can be packed, and is based on the space left by
large items in each configuration. The key point here is that the
linear program does not assign small items into specific windows
(located in specific bins), but only to types of windows. This new
method is the core of the improved running times of our schemes
over the running times of the previous results, which are only
asymptotic polynomial time approximation schemes (APTAS).
\end{abstract}

\section{Introduction}
\label{sec:intro} Classic bin packing
\cite{Ullman71,CsiWoe98,CoGaJo97,Gon32} is a natural and well
studied problem which has applications in problems of computer
storage, bandwidth allocation, stock cutting, transportation and
many other important fields. In the basic variant of this problem,
we are given $n$ items of size in $(0,1]$ which need to be
assigned to unit sized bins. Each bin may contain items of total
size at most 1, and the goal is to minimize the number of bins
used. Many variants of bin packing, coming from practical needs,
were studied \cite{Gon33,Gon34}. Below, we mention several such
variants that we address in this paper.

In various applications, such as storage of files on disks and
scheduling of jobs on bounded capacity processors, a bin can
contain only a limited number of items. This is the main
motivation for the study of the following variant of the bin
packing problem.  The {\sc bin packing problem with cardinality
constraints} (\pr) is defined as follows: The input consists of an
integer $k\geq 1$ (called the {\it cardinality bound}) and a set
of $n$ items $I=\{1,2,\ldots,n\}$, of sizes $1\geq s_1\geq s_2\geq
\cdots \geq s_n\geq 0$. The goal is to partition the items into
the minimum number of bins such that the total size of items in a
bin is at most 1, and the number of items in a bin is at most $k$.
The problem was introduced as early as in the 1970's by Krause,
Shen and Schwetman \cite{KSS75,KSS77}, and studied in a sequence
of papers \cite{CKP03,KP99,BCKK04,Epstein05}.

In other applications, such as bandwidth allocation, or any
application in storage that allows outsourcing, it is sometimes
possible to refuse to pack an item. This rejection clearly needs
to be compensated, and costs some given amount for each item.
This amount is called the {\it rejection penalty} of the item.
This situation can also occur in file servers where the items are
files, the bins are disks and rejection penalty of a file is the
cost of transferring it to be saved on an alternative media. In
other applications the rejection penalty refers to the penalty
caused by not serving a client. This motivates the following {\sc
bin packing problem with rejection} (\bpr).  The input to this
problem consists of $n$ items $I=\{ 1,2,\ldots ,n\}$, where item
$i$ has a size $s_i\in (0,1]$ and rejection penalty $r_i> 0$, and
the goal is to select a subset $A$ of $I$ (the set of accepted
items, i.e., items to be packed) and a partition of $A$ into
subsets $A_1,\ldots ,A_z$, for some integer $z$, such that
$\sum\limits_{i\in A_j} s_i \leq 1$ for all $j=1,2,\ldots ,z$, so
that $z+\sum\limits_{i\in I\setminus A} r_i$ is minimized. The
bin packing problem with rejection was introduced and studied by
He and D\'osa \cite{DoH05}. Further work on this problem can be
found in \cite{Eps06,BCH08}. Note that in this problem we assume
without loss of generality that all sizes of items are strictly
positive. If zero sized items exist, they could be all packed
into one
 bin, changing the cost of the solution by 1 in the case
 that all other items are rejected by this solution and without any change in the cost otherwise.

Note that both problems studied in the paper are generalizations
of classic bin packing. Classic bin packing is the special case of
\pr\ where $k=n$, and the special case of \bpr\ where all
rejection penalties are infinite.

For an algorithm ${\mathcal{A}}$, we denote its cost by
${\mathcal{A}}$ as well. The cost of an optimal offline algorithm
is denoted by \opt. We define the asymptotic approximation ratio
of an algorithm ${\mathcal{A}}$ is the infimum ${\mathcal{R}}\geq
1$ such that for any input, $ {\mathcal{A}} \leq{\mathcal{R}}
\cdot \mbox{\textsc{opt}} +c$, where $c$ is independent of the
input. If we enforce $c=0$, $\R$ is called the absolute
approximation ratio. An asymptotic polynomial time approximation
scheme is a family of approximation algorithms, such that for
every $\eps>0$ the family contains a polynomial time algorithm
with an asymptotic approximation ratio of $1+\eps$. We abbreviate
{\it asymptotic polynomial time approximation scheme} by APTAS
(also called an asymptotic PTAS). An asymptotic fully polynomial
time approximation scheme (AFPTAS) is an APTAS whose time
complexity is polynomial not only in the input size but also in
$1\over \eps$. If the scheme satisfies the definition above with
$c=0$, stronger results are obtained, namely, polynomial time
approximation schemes and a fully polynomial approximation
schemes, which are abbreviated as PTAS and FPTAS. Throughout the
paper, we use \opt\ to denote the cost of an optimal solution for
the original input, which is denoted by $I$, and we use \apx\ to
denote the cost of the solution returned by our schemes. For an
input $J$ we use $\opt(J)$ to denote the cost of an optimal
solution for the input $J$ (where $J$ is typically an adapted
input). Thus $\opt=\opt(I)$.

\noindent {\bf Previous results. \ } The classic bin packing
problem is known to admit an APTAS \cite{FerLue81} and an AFPTAS
\cite{KK82}.  Moreover, \pr\ and \bpr\ are problems that are known
to admit an APTAS \cite{CKP03,Eps06}. Specifically, the APTAS of
Caprara, Kellerer and Pferschy \cite{CKP03} generalizes the
methods of Fernandez de la Vega and Lueker \cite{FerLue81} where
items are rounded and grouped using linear grouping. The approach,
that is used to deal with items that are too small to be rounded
in this way, is not a greedy approach as in \cite{FerLue81} (which
fails in this case, as is demonstrated in \cite{CKP03}), but all
possible packings of large items are enumerated, and for each such
possibility, the small items are assigned to the already existing
bins and to new bins via a linear program. This last approach,
that is used also in \cite{Eps06}, where it is combined with
rounding of rejection penalties, results in an APTAS for each one
of the problems \pr\ and \bpr. Clearly, these schemes have large
running times due to enumeration. A different APTAS for \bpr, with
reduced running time, but that still uses some enumeration steps,
is given in \cite{BCH08}. Note that classic bin packing or any
generalization of it cannot be expected to admit an FPTAS or even
a PTAS, since approximating it within an absolute factor smaller
than $\frac 32$ is NP-hard (using a simple reduction from the {\sc
Partition} problem).

A problem which is dual to bin packing is the {\it bin covering}
problem. In this problem the goal is to maximize the number of
bins for which the total sizes of assigned items is at least 1.
This problem is known to admit an APTAS (by Csirik, Johnson and
Kenyon \cite{CJK01}) and an AFPTAS (by Jansen and  Solis-Oba
\cite{JS03}). The variant of this problem where a covered bin must
in addition to the constraint on the total size of items assigned
to it, must contain at least $k$ items, for a parameter $k$, we
get the problem of {\it bin covering with cardinality
constraints}. This last problem was considered in
\cite{EILbincover}, and was shown to admit an AFPTAS as well.  We
note that even though bin packing problems are the dual problems
of bin covering, the latter are maximization problems while the
former are minimization problems, and the nature of the problems
is very different. In particular, the methods of
\cite{EILbincover} as well as the methods of \cite{CJK01,JS03} are
not applicable for the problems considered in this paper, and
hence we needed to develop the new methods considered here.

We briefly survey the previous approaches used in the literature
for dealing with packing small items fractionally using a linear
program. One approach is to see all small items as fluid, without
distinguishing between different items. In this case, the small
items completely lose their identities, and the conversion process
of the fluid into items is performed in a later step, usually
greedily. The other approach which allows to keep the identities
of items simply allows to introduce constraints associated with
each small item separately, and the solution of the mathematical
program assigns each small item into a specific bin. The number of
constraints is linear in the number of items. It can be seen that
in the first approach there is no control whatsoever of the exact
allocation of specific items to specific bins, while the second
approach is rigid in the sense that all decisions must be made by
the mathematical program. Examples for the first approach can be
found as early as in the seminal work of Hochbaum and Shmoys
\cite{HS88} on scheduling uniformly related machines, in
approximation schemes for other scheduling problems (e.g.
\cite{AAWY98}), and in previous work on bin packing related
problems \cite{JS03,CJK01}. Examples of the second approach appear
in the previous work on the problems studied here
\cite{CKP03,Eps06}.

In this paper, we use a new and novel method of dealing with small
items which is an intermediate way between these two extreme
(previous) approaches. These items are packed using the
(approximated) solution of a linear program. To keep the running
time polynomial (in both the size of the input and $\eps$), the
linear program does not decide on the exact packing of these
items, but only on the type of a bin that they should join, where
a type of a bin is defined roughly according to the remaining size
in it after the packing of (rounded) large items, and in \pr, on
the cardinality constraint as well.


In this paper, we design an AFPTAS for each of the two problems,
\pr\ and \bpr. Studies of similar flavor was widely conducted for
other variants of bin packing and it is an established direction
of research, see e.g. \cite{JS05,Mur87,SY07}. The problems
studied in this paper were open in the sense that no AFPTAS is
known for them prior to this work.


We start the paper with Section \ref{methodss}, where we elaborate
on the methods used in this paper. In Sections \ref{pr} and
\ref{bpr} we describe the asymptotic approximation schemes for
\pr\ and \bpr, respectively, and prove their correctness. The
approximation scheme for \pr\ acts in two different ways
according to the value of $k$. The case of small $k$ is easier,
since in this case every bin has a relatively small number of
items. This simpler scheme is presented first, and it allows the
reader to get acquainted with the basic methods (but not with our
new methods for dealing with small items). The other schemes
presented in this paper, including the scheme for \pr\ with large
values of $k$, and the scheme for \bpr, require much more advanced
techniques. The methods that are employed, in order to obtain
these two schemes, are related, but each scheme requires different
specific ingredients, developed in this work, in order to solve
the problem it is meant for. Due to space limitations, we present
the easier case of Section \ref{pr}, and the AFPTAS for \bpr\ in
the Appendix.

Note that asymptotic fully polynomial approximation schemes,
unlike APTAS results, have practical running times especially if
one uses a method of solving approximately the linear program
which is faster than the ellipsoid method.  Such techniques are
available for packing problems.  Hence, AFPTAS can actually be
used to solve the problems they are developed for. Thus our
contribution, where we settle the complexity of the problems
studied here, is interesting not only from a theoretical point of
view, but also from the practical viewpoint.

\section{Methods}
\label{methodss} In this section we briefly describe the novel
methods that are used in this paper, side by side with adaptations
of well known methods, that are employed in this paper as well.

Linear grouping is a standard rounding method for bin packing
algorithms. It was first presented by Fernandez de la Vega and
Lueker \cite{FerLue81}. The main idea of this method is to round
the sizes of items into a small number of distinct sizes. Unlike
rounding methods for scheduling \cite{HS87}, the output sizes
resulting from the rounding must be representative sizes of real
item sizes in the input. In fact, this last method of rounding,
which rounds values from a given range (in our case, rejection
penalties) to a closest number among a fixed sized set of values,
such that no original value changes as a result by more than a
factor of $1+\eps$, is used in the paper as well in Section
\ref{bpr}. Both methods are typically used for sizes that are
large enough, where small sizes are treated separately. The
resulting set of large items has a small number of sizes
(usually, a function of $\eps$). On this set of sizes, valid
assignment configurations are defined, and a solution to a
packing problem is described as a set of bins packed according to
specific configurations.

In this paper, we introduce a new and novel method for treatment
of small items. These are typically items whose size is below some
threshold. Moreover, in Section \ref{bpr}, items of small
rejection penalty are seen as small as well. For the classic bin
packing problem, the treatment of such items is relatively easy.
After finding a solution, which is close enough to optimal
solution for the large items, small items can be added greedily
to the solution, using a simple packing heuristic such as Next
Fit or First Fit. As demonstrated by Caprara, Kellerer and
Pferschy \cite{CKP03}, using this approach for \pr\ leads to
approximation algorithms of high approximation ratio. That is,
finding an optimal packing of just the large items immediately
leads to poor performance for the complete input, no matter how
the small items are combined into the solution. In order to derive
an APTAS, i.e., an algorithm whose approximation ratio is within a
factor of $1+\eps$, but its running time is not necessarily
polynomial in $\frac 1{\eps}$, it is possible to enumerate a
large enough number of packings of the large items, and add the
small items to each one of them in some way (possibly in an
optimal fractional way). This results in a large number of
potential outputs, the best of which is actually given as the
output. On the other hand, if the goal is to design an AFPTAS,
where the running time must be polynomial in $\frac 1{\eps}$,
this approach cannot be successful. First, the number of packings
of large items is exponential in $\frac 1{\eps}$, and second, if
small items need to be added optimally and not greedily, it is
unclear whether it is possible to handle small items efficiently
since in their case, rounding of sizes is harmful. In summary, in
order to obtain an AFPTAS, we need a method which allows to
consider the small items in the linear program that seeks a
solution for the large items, but this linear program cannot
search for a complete and final solution for all items. We
introduce the new notion of {\it windows}. A window is a space in
which small items can be packed, and is based on the space left
by large items in each configuration. The key point here is that
the linear program does not assign small items into specific
windows (located in specific bins), but to types of windows.
Using rounding, we limit ourselves to a polynomial number of
window types. After the linear program has been solved, the small
items are assigned to specific windows based on the output. We
show that the items that cannot be assigned to windows (due to
fractional solutions, due to rounding, or due to the fact that
the windows reside in separate bins) can be packed separately
into new bins (or possibly rejected, in \bpr) with only a small
increase in the cost. In order to allow the packing of almost all
small items, in most cases studied here, their packing cannot
just be done greedily given the windows reserved for them. A
careful treatment that balances the load of small items in
similarly packed bins, taking into account not only the total
size of small items, but also their number, is required.

In order to find an AFPTAS and not an APTAS for each one of the
problems, the linear program, that is associated with the problem
on an adapted input (after rounding is applied, and our methods
for handling small items are invoked), cannot be solved exactly in
polynomial time (in $n$ and $\frac{1}{\eps}$). Therefore, we need
to solve it approximately. This is done via the column generation
technique of Karmarkar and Karp \cite{KK82}. The main idea is to
find an approximate solution of the dual linear program. For
that, we need to find a polynomial time separation oracle
(possibly, an approximate one) for each one of the dual programs.
This typically involves finding an approximate solution to a
knapsack type problem. The exact problem results from the exact
characteristics of the bin packing problem and the details of the
linear program used to solve it. We develop the linear programs
as well as the separation oracles for them in this paper.
Clearly, each problem results in a different linear program, a
different dual linear program, and a different separation oracle.
These separation oracles are based on application of fully
polynomial approximation schemes to variants of the knapsack
problem.

\section{An AFPTAS for \pr}
\label{pr} We fix a small value  $\eps < {1\over 2}$ such that
$1\over \eps$ is an integer. Our AFPTAS acts according to one of
two options, each of which is suitable for one case. In the first
case $k\leq {1\over \eps^2}$ and in the second case $k>{1\over
\eps^2}$. Recall that the set of items is denoted by $I$. We refer
to an item by its index, so $I=\{ 1,2,\ldots ,n\}$. The first
case can be found in the Appendix.

\subsection{Second case: $k>{1\over \eps^2}$}
We assume that $k < n$, otherwise there is no effective
cardinality constraint, and an AFPTAS for the problem follows from
the AFPTAS of Karmarkar and Karp \cite{KK82}. In this case we
partition the item set into {\it large items}, that is, items with
size at least $\eps$, and {\it small items} (all the non-large
items). We denote by $L$ the set of {\it large} items, and by
$S=I\setminus L$ the set of small items.  We apply linear grouping
to the large items (only).  That is, we partition the large items
into $1\over \eps^3$ classes $L_1,L_2,\ldots ,L_{1/\eps^3}$ such
that $\lceil |L| \eps^3 \rceil =|L_1| \geq |L_2| \geq \cdots \geq
|L_{1/\eps^3}|=\lfloor |L|\eps^3 \rfloor$, and such that if there
are two items $i,j$ with sizes $s_i>s_j$ and $i\in L_q$ and $j\in
L_p$ then $q \leq p$.  The two conditions uniquely define the
allocation of large items into classes up to the allocation of
equal-sized items. If $|L| < {1\over \eps^3}$, then instead of the
above partition, each large item has its own set $L_i$ such that
$L_1$ is an empty set, and for a large item $j$ we let $s'_j=s_j$
(i.e., we do not apply rounding in this case). We note that in
both cases (where we use the original partition of the items or
when $|L| < {1\over \eps^3}$) we have $|L_1| \leq 2\eps^3 |L|$.

Then, we round up the sizes of the items in $L_2,\ldots
,L_{1/\eps^3}$ as follows: For all values of $p=2,3,\ldots
,1/\eps^3$ and for each item $i\in L_p$, we let $s'_i=\max_{j\in
L_p} s_j$ to be the {\it rounded-up size of item $i$}.  For a
small item $i\in S$ we let $s'_i=s_i$.  The rounded-up instance
$I'$ consists of the set of items $I\setminus L_1$ where the size
of item $i$ is $s'_i$ for all $i$ ($k$ remains unchanged).  We
also define $L'=L\setminus L_1$ to be the set of large items in
the rounded-up instance. We have $\opt(I')\leq \opt$.

Given the rounded-up instance $I'$, we define  a {\it
configuration of large items of a bin} as a (possibly empty) set
of at most $k$ items of $L'$ whose total (rounded-up) size is at
most 1.  We denote the set of all configurations of large items by
$\cal C$. We denote the set of item sizes in $L'$ by $H$. For each
$v\in H$ we denote the number of items with size $v$ in $C$ by
$n(v,C)$, and  the number of items in $L'$ with size $v$ by $n(v)$
(where $n(v)=\lfloor |L| \eps^3 \rfloor$ or $n(v)=\lceil |L|
\eps^3 \rceil$, unless several classes are rounded to the same
size).

We denote the minimum non-zero size of an item by $s_{min}=
\min_{i\in S: s'_i\neq 0} s'_i$, and let
$$s'_{min}=\max\{\frac{1}{(1+\eps)^t}|t\in \mathbb{Z}, \
\frac{1}{(1+\eps)^t}\leq s_{min}\} \ . $$ We have $s'_{min}\leq
{s_{min}}$ and $s'_{min} > \frac{s_{min}}{1+\eps}$ and thus the
value $\log_{1+\eps}{\frac{1}{s'_{min}}}$ is polynomial in the
size of the input. We define the following set
$\W=\{(\frac{1}{(1+\eps)^t},a)|0\leq t \leq
\log_{1+\eps}{\frac{1}{s'_{min}}}+1, 0\leq a\leq k\}$. A {\it
window} is defined as a member of $\W$.  $\W$ is also called the
set of all possible windows. Then, $|{\cal W}| \leq n\cdot
(\log_{1+\eps} {1\over s'_{min}}+2) $, since the number of
possible values of the second component of a window is $k+1\leq
n$. For two windows, $w^1$ and $w^2$ where $w^i=(w^i_s,w^i_n)$
for $i=1,2$, we say that $w^1 \leq w^2$ if $w^1_n\leq w^2_n$ and
$w^1_s \leq w^2_s$.

Note that a bin, which is packed with large items according to
some configuration $C$, typically leaves space for small items.
For a configuration $C$, we denote the {\it main window of $C$} by
$w(C)=(w_s(C),w_n(C))$ (where $w_s(C)$ is interpreted as an
approximated available size for small items in a bin with
configuration $C$, and $w_n(C)$ is interpreted as the number of
small items that can be packed in a bin with configuration $C$
while still maintaining the cardinality constraint). More
precisely, a main window is defined as follows.  Assume that the
total (rounded-up) size of the items in $C$ is $s'(C)$. Let
$w_s(C)={1\over (1+\eps)^t}$ where $t$ is the maximum integer
such that $0 \leq t \leq \log_{1+\eps}{\frac{1}{s'_{min}}}+1$ and
that $s'(C) + {1\over (1+\eps)^t} \geq 1$, and $w_n(C)$ is the
difference between $k$ and  the number of large items in $C$
(i.e., $w_n(C)=k-\sum\limits_{v\in H} n(v,C)$). Note that in any
case, for a non-trivial input, $t$ can take at least three values.
The main window of a configuration is a window (i.e., belongs to
$\W$), but $\W$ may include windows that are not the main window
of any configuration. We note that $|{\cal W}|$ is polynomial in
the input size and in $1\over \eps$, whereas $|{\cal C}|$ may be
exponential in $1\over \eps$ (specifically, $|{\cal C}| \leq
({1\over \eps^3}+1)^{1/\eps}$, since in configuration there are
up to $1\over \eps$ large items of $|H|\leq{ 1\over \eps^3}$
sizes).  We denote the set of windows that are actual main windows
of at least one configuration by ${\cal {W}'}$. We first define a
linear program that allows the usage of any window in $\W$. After
we obtain a solution to this linear program, we modify it  so
that it only uses windows of ${\cal W}'$.

We define a {\it generalized configuration} $\tilde{C}$ as a pair
$\tilde{C}=(C,w=(w_s,w_n))$, for some configuration $C$ and some
$w\in \W$. A generalized configuration $\tilde{C}$ is a {\it
valid generalized configuration} if $w\leq w(C)$. The set of all
valid generalized configurations is denoted by $\tilde{\C}$.

For $W\in {\cal W}$ denote by $C(W)$ the set of generalized
configurations $\tilde{C}=(C,w=(w_s,w_n))$, such that $w=W$.
That is, $C(W)=\{\tilde{C}=(C,w)\in {\tilde{\cal C}} : W =w\}$.

We next consider the following linear program.   In this linear
program we have a variable $x_{\tilde{C}}$ denoting the number of
bins with generalized configuration $\tilde{C}$, and variables
$Y_{i,W}$ indicating if the small item $i$ is packed in a window
of type $W$ (the exact instance of this window is not specified in
the solution of the linear program).

\begin{eqnarray}
\min & \sum\limits_{\tilde{C}\in {\tilde{\cal C}}} x_{\tilde{C}}& \nonumber \\
s.t.& \sum\limits_{\tilde{C}=(C,w)\in {\tilde{\cal C}}} n(v,C)
x_{\tilde{C}} \geq n(v) &
\forall v\in H\label{cons1}\\
& \sum\limits_{W\in {\cal W}} Y_{i,W} \geq 1& \forall i\in S\label{cons2}\\
& w_s\cdot \sum\limits_{{\tilde{C}}\in C(W)} x_{\tilde{C}} \geq \sum\limits_{i\in S} s'_i\cdot Y_{i,W} & \forall W=(w_s,w_n)\in {\cal W} \label{cons3} \\
& w_n\cdot \sum\limits_{{\tilde{C}}\in C(W)} x_{\tilde{C}}\geq \sum\limits_{i\in S} Y_{i,W} & \forall W=(w_s,w_n)\in {\cal W}\label{cons4}\\
& x_{\tilde{C}} \geq 0& \forall {\tilde{C}}\in {\tilde{\cal
C}},\nonumber
\\ & Y_{i,W} \geq 0&
 \forall W\in {\cal W}, \forall i\in S.\nonumber
\end{eqnarray}
Constraints (\ref{cons1}) and (\ref{cons2}) ensure that each item
(large or small) of $I'$ will be considered. The large items will
be packed by the solution, and the small items would be assigned
to some type of window (but not to a specific location).
Constraints (\ref{cons3}) ensure that the total size of the small
items that we decide to pack in window of type $W$ is not larger
than the total available size in all the bins that are packed
according to a generalized configuration, whose second component
is a window of type $W$. Similarly, the family of constraints
(\ref{cons4}) ensures that the total number of the small items
that we decide to pack in a window of type $W$ is not larger than
the total number of small items that can be packed (in accord with
the cardinality constraint) in all the bins whose generalized
configuration of large items induces a window of type $W$. The
linear relaxation assumes, in particular, that small items can be
assigned fractionally to windows, that is, the small items leave
no gaps. We later show how to construct a valid allocation of
small items, leaving a small enough number of small items, whose
total size is small enough as well, unpacked. We further show how
to deal with these unpacked items. We note that $\opt(I')$ implies
a feasible solution to the above linear program that has the cost
$\opt(I')$, since the packing of the small items clearly satisfies
the constraints (\ref{cons3}) and (\ref{cons4}), and the packing
of large items satisfies the constraints (\ref{cons1}). Moreover,
it implies a solution to the linear program in which all variables
$x_{\tilde{C}}$, that correspond to generalized configurations
$\tilde{C}=(C,w)$ for which $w$ is not the main window of $C$, are
equal to zero, and all variables $Y_{i,w}$ where $w \notin \W'$
are equal to zero as well.

We invoke the column generation technique of Karmarkar and Karp
\cite{KK82} as follows.  The above linear program has an
exponential number of variables and a polynomial number of
constraints (neglecting the non-negativity constraints). Instead
of solving the linear program we solve its dual program (that has
a polynomial  number of variables and an exponential number of
constraints). The variables $\alpha_v$ correspond to the item
sizes in $H$. The variables $\beta_i$ correspond to the small
items. The intuitive meaning of the variables can be seen as
weights assigned to these items.  For each $W\in {\cal W}$ we have
a pair of dual variables $\gamma_W$ and $\delta_W$.  Using these
dual variables, the dual linear program is as follows.

\begin{eqnarray}
\max & \sum\limits_{v \in H} n(v) \alpha_v+ \sum\limits_{i\in S} \beta_i&\nonumber \\
s.t. & \sum\limits_{v\in H} n(v,C) \alpha_v +
 w_s\gamma_w+w_n\delta_w
\leq 1& \forall \tilde{C}=(C,w=(w_s,w_n))\in {\tilde{\cal C}}\label{con1}\\
& \beta_i - s'_i\gamma_W-\delta_W \leq 0&\forall i\in S, \ \forall W\in {\cal W}\label{con2}\\
& \alpha_v \geq 0& \forall v\in H \nonumber \\& \beta_i \geq 0 &
\forall i \in S \nonumber \\ & \gamma_W,\delta_W \geq 0& \forall
W\in {\cal W}.\nonumber
\end{eqnarray}

First note that there is a polynomial number of constraints of
type (\ref{con2}), and therefore we clearly have a polynomial time
separation oracle for these constraints.  If we would like to
solve the above dual linear program (exactly) then using the
ellipsoid method we need to establish the existence of a
polynomial time separation oracle for the constraints
(\ref{con1}).  However, we are willing to settle on an
approximated solution to this dual program. To be able to apply
the ellipsoid algorithm, in order to solve the above dual problem
within a factor of $1+\eps$, it suffices to show that there exists
a polynomial time algorithm (polynomial in $n$, $1\over \eps$ and
$\log \frac 1{s'_{min}}$) such that for a given solution
$a^*=(\alpha^*,\beta^*,\gamma^*,\delta^*)$ decides whether $a^*$
is a feasible dual solution (approximately). That is, it either
provides a generalized configuration $\tilde{C}=(C,w=(w_s,w_n))
\in \tilde{\C}$ for which $\sum\limits_{v\in H} n(v,C) \alpha^*_v
+ w_s\gamma^*_w+w_n \delta^*_w> 1$, or outputs that an approximate
infeasibility evidence does not exist, that is, for all
generalized configurations $\tilde{C}=(C,w=(w_s,w_n)) \in
\tilde{\C}$, $\sum\limits_{v\in H} n(v,C) \alpha^*_v +
w_s\gamma^*_w+w_n \delta^*_w \leq 1+\eps$ holds. In such a case,
$a^*\over 1+\eps$ is a feasible dual solution that can be used.
 Such a configuration $\tilde{C}$ can be found
by the following procedure:  For each $W=(w_s,w_n) \in {\cal W}$
we look for a configuration $C \in \C$ such that $(C,W)$ is a
valid generalized configuration, and $\sum\limits_{v\in H} n(v,C)
\alpha^*_v$ is maximized. If a configuration $C$ that is indeed
found, the generalized configuration, whose constraint is checked,
is $(C,W)$. To find $C$, we invoke an FPTAS for the KCC problem
with the following input: The set of items is $H$ where for each
$v\in H$ there is a volume $\alpha^*_v$ and a size $v$, the goal
is to pack a multiset of the items,  so that the total volume is
maximized, under the following conditions. The multiset should
consist of at most $k-w_n$ items (taking the multiplicity into
account, but an item can appear at most a given number of times).
The total (rounded-up) size of the multiset should be smaller than
$1-\frac{w_s}{1+\eps}$, unless $w_s<s'_{min}$, where the total
size should be at most 1 (in this case, the window leaves space
only for items of size zero). Since the number of applications of
the FPTAS for the KCC problem is polynomial (i.e., $|{\cal W}|$),
this algorithm runs in polynomial time.  If it finds a solution,
that is, a configuration $C$, with total volume greater than
$1-w_s\gamma^*_W-w_n \delta^*_W$, we argue that $(C,W)$ is indeed
a valid generalized configuration, and this implies that there
exists a generalized configuration, whose dual constraint
(\ref{con1}) is violated. By the definition of windows, the
property $w_s<s'_{min}$ is equivalent to
$w_s=\frac{s'_{min}}{1+\eps}$, which is the smallest size of
window (and the smallest sized window forms a valid generalized
configuration with any configuration, provided that the value of
$w_n$ is small enough). Since $C$ has at most $k-w_n$ items, the
main window of $C$ in this case is no smaller than $(w_s,w_n)$ and
therefore, the generalized configuration $(C,W)$ is valid. If
$w_s\geq s'_{min}$, recall that the main window of $C$,
$(w_s(C),w_n(C))$ is chosen so that
$s'(C)+{w_s(C)} \geq 1$, and that $C$ is chosen by the algorithm
for KCC so that $s'(C)<1-\frac{w_s}{1+\eps}$. We get $1-w_s(C)
\leq s'(C)<1-\frac{w_s}{1+\eps}$ and therefore $w_s <
(1+\eps)w_s(C)$, i.e.,  $w_s \leq w_s(C)$ (since the sizes of
windows are integer powers of $1+\eps$). Since $C$ contains at
most $k-w_n$ items, we have  $w_n(C) \geq w_n$ and so we conclude
that $W \leq (w_s(C),w_n(C))$, and $(C,W)$ is  a valid generalized
configuration.
Thus in this case we found that this solution is a configuration
whose constraint in the dual linear program is not satisfied, and
we can continue with the application of the ellipsoid algorithm.

Otherwise, for any window $W=(w_s,w_n)$ and any configuration $C$
of total rounded-up size less than $1-\frac{w_s}{1+\eps}$ (or at
most 1, if $w_s<s'_{min}$), with at most $k-w_n$ items, the total
 volume is at most $(1+\eps)(1-w_s\gamma^*_W-w_n \delta^*_W)\leq
(1+\eps)-w_s\gamma^*_W-w_n \delta^*_W $. We prove that in this
case, all the constraints of the dual linear program are
satisfied by the solution $a^*\over 1+\eps$. Consider a valid
generalized configuration
$\tilde{C}=(C,(\tilde{w}_s,\tilde{w}_n))$. We have
$(\tilde{w}_s,\tilde{w}_n)\leq (w_s(C),w_n(C))$, where
$(w_s(C),w_n(C))$ is the main window of $C$. If $w_s(C)<s'_{min}$,
then $\tilde{w}_s=w_s(C)$. Since $s'(C)\leq 1$ for any
configuration, and $\tilde{w}_n\leq w_n(C)$, where $k-w_n(C)$ is
the number of items in $C$, $C$ is a possible configuration to be
used with the window $(\tilde{w}_s,\tilde{w}_n)$ in the
application of the FPTAS for KCC. Assume next that $\tilde{w}_s
<1$, then when the FPTAS for KCC is applied on
$W=(\tilde{w}_s,\tilde{w}_n)$, $C$ is a configuration that is
taken into account for $W$ since
$s'(C)<1-\frac{w_s(C)}{1+\eps}\leq 1-\frac{\tilde{w}_s}{1+\eps}$,
where the first inequality holds by definition of $w_s(C)$, and
$C$ has at most $k-w_n(C) \leq k-\tilde{w}_n$ items. If
$\tilde{w}_s =1$ then $1 \geq w_s(C) \geq \tilde{w}_s =1$, so
$w_s(C)=1$. A configuration $C_1$ that contains at least one large
item satisfies $s'(C_1)\geq \eps$, so $s'(C_1)+\frac{1}{1+\eps}
\geq \frac{1+\eps+\eps^2}{1+\eps}>1$. Therefore if the main window
of a configuration is of size 1, this configuration is empty. We
therefore have that $C$ is an empty configuration, thus $s'(C)=0$
and $\tilde{w}_n \leq w_n(C)=k$. This empty configuration $C$ is
considered with any possible window.

We denote by $(x^*,Y^*)$ the solution to the primal linear program
that we obtained. Since its cost is a $(1+\eps)$ approximation for
the optimal solution to the linear program, we conclude that
$\sum\limits_{\tilde{C}\in {\tilde{\cal C}}}x^*_{\tilde{C}} \leq
(1+\eps)\opt(I')$.

We modify the solution to the primal linear program, into a
different feasible solution of the linear program, without
increasing the objective function. We create a list of generalized
configurations whose $x^*$ component is positive.
From this list of generalized configurations, we find a list of
windows that are the main window of at least one  configuration
induced by a generalized configuration in the list. This list of
windows is a subset of $\cal{W}'$ defined above. We would like the
solution to use only windows from $\cal{W}'$.

The new solution will have the property that any non-zero
components of $x^*$, $x^*_{\tilde{C}}$ corresponds to a
generalized configuration $\tilde{C}=(C,w)$, such that $w\in
\cal{W}'$. We still allow generalized configurations
$\tilde{C}=(C,w)$ where $w$ is not the main window of $C$, as long
as $w\in \W'$. This is done in the following way. Given a window
$w' \notin \W'$, we define $X_{w'}=\sum\limits_{\tilde{C''}\in
C(w')}x^*_{\tilde{C''}}$. The following is done in parallel for
every generalized configuration $\tilde{C'}=(C,w')$, where
$w'\notin \W'$ and such that $x^*_{\tilde{C'}}>0$, where the main
window of $C$ is $w \geq w'$ (but $w'\neq w$). We let
$\tilde{C}=(C,w)$. The windows allocated for small items need to
be modified first, thus an amount of
$\frac{x^*_{\tilde{C'}}}{X_{w'}}Y_{i,w'}$ is transferred from
$Y_{i,w'}$ to $Y_{i,w}$. We modify the values $x^*_{\tilde{C'}}$
and $x^*_{\tilde{C}}$ as follows. We increase the value of
$x^*_{\tilde{C}}$ by an additive factor of $x^*_{\tilde{C'}}$ and
let $x^*_{\tilde{C'}}=0$.

To show that the new vector $(x^*,Y^*)$ still gives a feasible
solution of the same value of objective function, we consider the
modifications. The sum of components of $x^*$ does not change at
all in the above process, thus the value of the objective function
is the same. Moreover, for every configuration $C$, the sum of
components $x^*$, that correspond to  generalized configurations
whose configuration of large items is $C$, does not change. Thus
the constraints (\ref{cons1}) still hold. We next consider the
constraint (\ref{cons2}) for $i$, for a given small item $i\in S$.
Since the sum of variables $Y^*_{i,W}$ does not change, this
constraint still holds.

As for constraints (\ref{cons3}) and (\ref{cons4}), for a window
$w \notin \cal{W'}$, the right hand side of each such constraint
became zero. On the other hand, for windows in $\cal{W'}$, every
increase in some variable $x^*_{\tilde{C}}$ for
$\tilde{C}=(C,w=(w_s,w_n))$, that is originated in a decrease of
$x^*_{\tilde{C'}}$ for $\tilde{C'}=(C,w'=(w'_s,w'_n)$ is
accompanied with an increase of
$\frac{x^*_{\tilde{C'}}}{\sum\limits_{\tilde{C''}\in
C(w')}x^*_{\tilde{C''}}}Y^*_{i,w'}=\frac{x^*_{\tilde{C'}}}{X_{w'}}Y^*_{i,w'}
$ in $Y^*_{i,w}$, for every $i\in S$, that is, an increase of
$\sum\limits_{i\in S}\frac{x^*_{\tilde{C'}}}{X_{w'}} s'_i\cdot
Y^*_{i,w'}$ in the right hand size of the constraint (\ref{cons3})
for $w$, and an increase of $w_s \cdot x^*_{\tilde{C'}}$ in the
left hand side. Since we have $w_s \cdot X_{w'} \geq w'_s \cdot
X_{w'} \geq \sum\limits_{i\in S} s'_i\cdot Y^*_{i,w'}$ before the
modification occurs (since constraint (\ref{cons3})  for the
window $w'$ holds for the solution before the modification), we
get that the increase of the left hand side is no smaller than the
increase in the right hand side. There is an increase of
$\sum\limits_{i\in S}\frac{x^*_{\tilde{C'}}}{X_{w'}} \cdot
Y_{i,w'}$ in the right hand size of the constraint (\ref{cons4})
for $w$, and an increase of $w_n \cdot x^*_{\tilde{C'}}$ in the
left hand side. Since we have $w_n \cdot X_{w'} \geq w'_n \cdot
X_{w'} \geq \sum\limits_{i\in S} Y^*_{i,w'}$, we get that the
increase of the left hand side is no smaller than the increase in
the right hand side.

Now, we can delete the constraints of (\ref{cons3}) and
(\ref{cons4}) that correspond to windows in ${\cal W} \setminus
{\cal W}'$.  In the resulting linear program we consider a basic
solution that is not worse than the solution we obtained above.
Such  a basic solution can be found in polynomial time.  We denote
this basic solution by $(\x^*,\Y^*)$.

We apply several steps of rounding to obtain a feasible packing
of the items into bins.  We first round up $\x^*$.  That is,
denote by $\hat{x}$ the vector such that
$\hat{x}_{\tilde{C}}=\lceil \x^*_{\tilde{C}} \rceil$ for all
$\tilde{C}\in {\tilde{\cal C}}$. Moreover, each small item $i\in
S$ such that $(\Y^*_{i,W})_{W\in {\cal W}}$ is fractional, is
packed using a dedicated bin. We modify the value of
$\hat{x}_{\tilde{C}}$ for $\tilde{C}$ that corresponds to an
empty configuration $C$, together with the window $(1,k)$, to
reflect the additional bins that accommodate the small items that
were previously packed fractionally. We modify the values
$(\Y^*_{i,W})_{W\in {\cal W}}$ so that every item $i$ which is
packed into a new bin has $\hat{Y}_{i,W}=0$ for all $W$, except
for $W=(1,k)$ for which $\hat{Y}_{i,W}=1$. For all other
variables $Y_{i,W}$ we define $\hat{Y}_{i,W}=\Y^*_{i,W}$. We next
bound the increase in the cost due to this rounding.

\begin{lemma}\label{l1}
$\sum\limits_{\tilde{C}\in {\tilde{\cal C}}}
\hat{x}_{\tilde{C}}\leq \sum\limits_{\tilde{C}\in {\tilde{\cal
C}}}\x^*_{\tilde{C}} + |H|+2|{\cal W'}|$.
\end{lemma}
\begin{proof}
Consider now the primal linear program, where constraints
(\ref{cons3}) and (\ref{cons4}) exist only for windows in $\W'$,
the variables $x_{\tilde{C}}$ exist only for generalized
configurations $\tilde{C}=(C,w)$ where $w \in \W'$, and the
variables $Y_{i,W}$ exists only for $W\in\W'$. The basic solution
$(\x^*,\Y^*)$ is a feasible solution for this linear program. In
the primal linear program there are $|H|+2|{\cal W'}|+|S|$
inequality constraints, and hence in a basic solution there are
at most $|H|+2|{\cal W'}|+|S|$ basic variables. For every $i\in
S$, there is at least one $W'$ such that $Y_{i,W'}$ is a basic
variable, and therefore the number of basic variables from the $x$
components and additional basic variables from the $Y$ components
is at most $|H|+2|{\cal W'}|$. Hence the sum of the number of
fractional components among the $x_{\tilde{C}}$ variables, and
the number of small items such that the vector $(\Y^*_{i,W})_W$
contains more than one non-zero component is at most $|H|+2|{\cal
W'}|$. This is an upper bound on the difference in the objective
values of the two solutions and the claim follows.
\end{proof}


Our scheme returns a solution that packs $\hat{x}_{\tilde{C}}$
bins with configuration $\tilde{C}$. Each large item of the
rounded-up instance is replaced by the corresponding item of
$I$.  We clearly use at most $\sum\limits_{C\in {\cal C}}
\hat{x}_C$ bins in this way.
 We
next pack each item of $L_1$ by its own bin (if $L_1$ is
non-empty). We denote the resulting solution by $SOL_{large}$.

\begin{lemma}
The cost of $SOL_{large}$ is at most $\sum\limits_{\tilde{C}\in
{\tilde \C}}\hat{x}_{\tilde{C}} + \eps\opt$.
\end{lemma}
\begin{proof}
It suffices to show that $|L_1| \leq \eps\opt$. To see this last
claim note that $|L_1|\leq 2|L|\eps^3$ and each item in $L$ has
size at least $\eps$ and therefore $\opt \geq |L|\eps$, and
therefore $|L_1| \leq 2\eps^2 \opt$ and the claim follows since
$\eps <{1/2}$.
\end{proof}

\begin{corollary}
The number of bins used by $SOL_{large}$  is at most
$(1+2\eps)\opt + |H|+2|{\cal W'}|$.
\end{corollary}

We next consider the packing of the small items that are supposed
to be packed (according to $\hat{Y}$) in bins with a window of
type $W=(w_s,w_n)$. Assume that there are $X(W)$ such bins (i.e.,
$X(W)=\sum\limits_{\tilde{C}:\tilde{C}=(C,W)}
\hat{x}_{\tilde{C}}$). Denote by $S(W)$ the set of small items
that we decided to pack in bins with window $W$ (for some of
these items we will change this decision in the sequel).  Then,
by the feasibility of the linear program we conclude that
$|S(W)|\leq w_n X(W)$ and $\sum\limits_{i\in S(W)} s'_i \leq w_s
X(W)$.  We next show how to allocate almost all items of $S(W)$
to the $X(W)$ bins with window $W$ such that each such bin will
contain  at most $k$ items (that is, at most $w_n$ small items)
and the total size of items in each such bin will be at most
$1+\frac{\eps w_s}{1+\eps}$.  To do so, we sort the items in
$S(W)$ according to non-increasing size (assume the sorted list
of item indices is $b_1 \geq b_2 \geq  \ldots \geq b_{|S(W)|}$).
Then, allocate the items to the bins in a round-robin manner, so
that bin $j$ receives items of indices $b_{j+p\cdot X(W)}$ for
all integers $p\geq 0$ such that $j+p\cdot X(W) \leq |S(W)|$. We
call the allocation of items for a given value of $p$ a {\it
round of allocations}. If $w_s=\frac{s'_{min}}{1+\eps}$ then all
items assigned to this type of window are of size 0, and the
resulting allocation is valid in the sense that every bin
contains items of total size at most 1, and at most $k$ items per
bin, and there is no need to adapt the packing of small items. We
therefore assume $w_s \geq s'_{min}$. We claim that the last bin
of index $X(W)$ received at most an $\frac{1}{X(W)}$ fraction of
the total size of the items, whose sum is equal to
$\sum\limits_{i=1}^{|S(W)|} b_i$. To prove this, we artificially
add at most $X(W)-1$ items of size zero to the end of the list
(these items are added just for the sake of the proof), and
allocate them to the bins that previously did not receive an item
in the last round of allocations, that is, bins $r,\ldots,X(W)$
such that bin $r-1<X(W)$ received the last item. If bin $X(W)$
received the last item then no items are added. Now the total
size of small items remained the same, but every bin got exactly
one item in each round. Since the last bin received the smallest
item in each round, the claim follows. On the other hand, we can
apply the following process, at every time $i < X(W)$, remove the
first (largest) small item from bin $i$. As a result, the
round-robin assignment now starts from bin $i+1$ and bin $i$
becomes the bin that receives items last in every round, and thus
by the previous proof, the total size of items assigned to it is
at most $\frac{\sum\limits_{i=1}^{|S(W)|} b_i}{X(W)}$ (since the
total size of items does not increase in each step).

We create an intermediate solution $SOL_{inter}$ by removing the
largest small item from every bin and packing these removed items
in separate bins in groups of $1\over \eps$ removed items per bin
(note that such bin is feasible as the total size of $1\over \eps$
small items is at most 1 and $k > {1\over \eps}$ and hence the
cardinality constraint is satisfied as well).  The total cost of
this intermediate solution is therefore at most $(1+\eps)\cdot
\left( (1+2\eps)\opt + |H|+2|{\cal W'}| \right) +1$ (the last bin
can contain less than $1\over \eps$ such removed items).  We note
that since we divide the items in $S(W)$ equally (up to a
difference of one item) to the $X(W)$ bins, we conclude that after
the removal of one item from each bin (or even before that), every
such bin has at most $k$ items (both large and small), and
therefore all the bins satisfy the cardinality constraint.
Moreover, the total size of small items assigned to such bin
(after the removal of one item per bin) is at most $w_s$ by the
above argument regarding the total size of small items in a bin
where the largest small item was removed.

The intermediate solution is infeasible because our definition of
$w_s$ is larger than the available space for small items in such
bin. We create the final solution $SOL_{final}$ as follows.
Consider a bin such that the intermediate solution packs to it
large items according to configuration $C$, and small items with
total size at most $w_s(C)$.  For every bin, we do not change the
packing of large items. As for the small items, we remove them
from the bin and start packing the small items into this bin
greedily in non-decreasing order of the item sizes, as long as the
total size of items packed to the bin does not exceed 1. The first
item that does not fit into the bin we pack in separate bins (each
such separate bin will contain $1\over \eps$ such first items for
different bins of $SOL_{inter}$). Similarly to the above argument
these are feasible bins and they add an additive factor of $\eps$
times the cost of $SOL_{inter}$ to the total cost of the packing
(plus 1).

By the definition of windows, the actual space in a bin with
window $(w_s,w_n)$, that is free for the use of small items, is
at least of size $\frac{w_s}{1+\eps}$. After the removal of the
first item that does not fit into the space for small items, the
remaining small items allocated to this bin have a total size of
at most $w_s-\frac{w_s}{1+\eps}=\eps \frac{w_s}{1+\eps}$. Since by
definition, $\frac{w_s}{1+\eps}<1$, the small items that were
assigned to a bin but cannot be packed (not including the first
item that was packed in the previous step) are of total size at
most $\eps$. Similar considerations can be applied to the
cardinality of these items. Since the unpacked items are the
largest ones, the remaining unpacked items in a bin of
$SOL_{inter}$ have cardinality of at most $\eps w_n \leq \eps k$.
Therefore, we can pack the unpacked items of every $1\over \eps$
bins of $SOL_{inter}$ using one additional bin. In this way we get
our final solution $SOL_{final}$.  We note that the cost of
$SOL_{final}$ is at most $(1+2\eps)$ times the cost of
$SOL_{inter}$ plus two.  Therefore the cost of $SOL_{final}$ is at
most $ (1+2\eps) \left( (1+\eps)\cdot \left( (1+2\eps)\opt +
|H|+2|{\cal W'}| \right) +1 \right) +2 \leq (1+10\eps)\opt +
5(|H|+|{\cal W'}|+1) \leq (1+10\eps)\opt + 5({1\over \eps^3}+
({1\over \eps^3}+1)^{1/\eps}) +2$ where the last inequality holds
by $\eps < {1\over 2}$, $|H|\leq {1\over \eps^3}$ and $|{\cal
W}'|\leq |{\cal C}| \leq ({1\over \eps^3}+1)^{1/\eps}$. Therefore,
we have established the following theorem.

\begin{theorem}
If $k\geq {1\over \eps^2}$, the above scheme is an AFPTAS for \pr.
\end{theorem}

Since we covered both cases, we obtain the following.
\begin{theorem}
The above scheme is an AFPTAS for \pr.
\end{theorem}


\small
\bibliographystyle{plain}


\newpage
\appendix

\section{First case of Section \ref{pr}: $k\leq {1\over \eps^2}$} In this case we
apply linear grouping for {\it all} items, and do not classify
items into types. That is, we partition the items into $1\over
\eps^3$ classes $L_1,L_2,\ldots ,L_{1/\eps^3}$ such that $\lceil
n\eps^3 \rceil =|L_1| \geq |L_2| \geq \cdots \geq
|L_{1/\eps^3}|=\lfloor n\eps^3 \rfloor$ (note that this condition
uniquely identifies the cardinality of each class), and such that
if there are two items $i,j$ with sizes $s_i>s_j$ and $i\in L_q$
and $j\in L_p$ then $q \leq p$ (so $L_1$ receives the subset of
largest items, and for $2\leq p \leq {1\over \eps^3}$, $L_p$
  receives the largest items
from $I\setminus \left( L_1\cup \cdots \cup L_{p-1}\right)$).  The
two conditions uniquely define the allocation of items into
classes up to the allocation of equal sized items.

Then, we round up the sizes of the items in $L_2,\ldots
,L_{1/\eps^3}$ as follows: For all values of $p$, $p=2,3,\ldots
,1/\eps^3$, and for each item $i\in L_p$, we define
$s'_i=\max_{j\in L_p} s_j$ to be the {\it rounded-up size of item
$i$}.  The rounded-up instance $I'$ consists of the set of items
$I\setminus L_1$, where for every $i$, the size of item $i$ is
$s'_i$ (the parameter $k$ remained unchanged). We next argue that
$\opt(I') \leq \opt$. Given an optimal solution to $I$, \opt, we
transform it into a solution to $I'$. We define a bijection from
$I'$ to $I$, so that every item of $I'$ is mapped to an item of
$I$ that is no smaller, and can take its place in the packing.
Since $I'$ does not contain $L_1$, and since $|L_i|\leq
|L_{i-1}|$ (in both $I$ and $I'$, since the size of sets is not
influenced by the rounding), we map every item of $L_i$ (for all
$i\geq 2$) in $I'$ to some item of $L_{i-1}$ in $I$. By our
rounding, every item of $L_i$ in $I'$ is no larger than any item
of $L_{i-1}$ in $I$.

Given the rounded-up instance $I'$, we let a configuration of a
bin $C$ be a set of at most $k$ items of $I'$ whose total
(rounded-up) size is at most 1.  We denote the set of all
configurations by $\cal C$ (this set is not computed explicitly,
and typically has an exponential size). We denote the set of item
sizes in $I'$ by $H$. For each $v\in H$, we let $n(v,C)$ be the
number of items with size $v$ in $C$, and we let $n(v)$ be the
number of items in $I'$, with size $v$ (where $n(v)=\lfloor
n\eps^3 \rfloor$ or $n(v)=\lceil n\eps^3 \rceil$, unless several
classes are rounded to the same size, thus $|H|\leq
\frac{1}{\eps^3}$). We solve (approximately) the following linear
program, where for each configuration $C$, there is a variable
$x_C$ indicating the number of bins packed using configuration
$C$.

\begin{eqnarray*}
\min & \sum\limits_{C\in {\cal C}} x_C& \\
s.t.& \sum\limits_{C\in {\cal C}} n(v,C) x_C \geq n(v) & \forall v\in H\\
& x_C \geq 0& \forall C\in {\cal C}.
\end{eqnarray*}
We let $x^*$ be an approximate (within a factor of $1+\eps$)
solution to this linear program, and further define $y_C=\lceil
x^*_C \rceil$, for every configuration $C \in \C$. It can be seen
that the vector $y$ is a feasible solution to the linear program
(since $y_C \geq x_C$ for all $C$, it satisfies all the
constraints). Our scheme returns a solution that packs $y_C$ bins
with configuration $C$, in this solution, there are at least
$n(v)$ slots for every $v \in H$. Note that some of these slots
may remain empty, which happens in the case that the number of
slots is strictly larger than $n(v)$. To get a solution for
$I\setminus L_1$, each item of the rounded-up instance is replaced
by the corresponding item of $I$.  We clearly use at most
$\sum\limits_{C\in {\cal C}} y_C$ bins in this way.

To solve the above linear program approximately, we invoke the
column generation technique of Karmarkar and Karp \cite{KK82}. We
next elaborate on this technique.  The above linear program has
an exponential number of variables and a polynomial number of
constraints (neglecting the non-negativity constraints). Instead
of solving the linear program, we solve its dual program (that has
a polynomial number of variables and an exponential number of
constraints). The variables $z_v$ correspond to the item sizes in
$H$, their intuitive meaning can be seen as weights of these
items.
\begin{eqnarray*}
\max & \sum\limits_{v \in H} n(v) z_v&\\
s.t. & \sum\limits_{v\in H} n(v,C) z_v \leq 1& \forall C\in {\cal C}\\
& z_v \geq 0& \forall v\in H.
\end{eqnarray*}
To be able to apply the ellipsoid algorithm, in order to solve the
above dual problem within a factor of $1+\eps$, it suffices to
show that there exists a polynomial time algorithm (polynomial in
$n$ and $1\over \eps$) such that for a given solution $z^*$ (which
is a vector of length $|H|\leq \frac 1{\eps^3}$), decides whether
$z^*$ is a feasible dual solution (approximately). That is, it
either provides a configuration $C \in \C$ such that
$\sum\limits_{v\in H} n(v,C) z^*_v > 1$, or outputs that an
approximate infeasibility evidence does not exist, that is, for
all configurations $C \in \C$, $\sum\limits_{v\in H} n(v,C) z^*_v
\leq 1+\eps$ holds. In such a case, $z^*\over 1+\eps$ is a
feasible dual solution that can be used. Such a configuration $C$
can be found using an FPTAS for the following {\sc knapsack
problem with a maximum cardinality constraint} (KCC): Given a set
of item types $H$, where each item type $v \in H$ has a given
multiplicity $n(v)$, a volume $z^*_v$ and a size $v$, the goal is
to pack a multiset of at most $k$ items (taking the multiplicity,
in which items are taken, into account, and letting the solution
contain at most $n(v)$ items of type $v$) and a total size of at
most 1, so that the total volume is maximized. If the FPTAS to KCC
finds a solution with a total volume greater than 1, then this
solution is a configuration whose constraint in the dual linear
program is violated, and we can continue with the application of
the ellipsoid algorithm. Otherwise, the FPTAS to KCC finds a
solution with a total volume of at most 1, since the FPTAS is an
$1+\eps$ approximation, it means that no solution with a total
volume larger than $1+\eps$ exists, and therefore, all the
constraints of the dual linear program are satisfied by the
solution $z^*\over 1+\eps$. To provide an FPTAS for KCC, note
that one can replace an item with size $v$ by $n(v)$ copies of
this item and then one can apply the FPTAS of Caprara et al.
\cite{CKPP} for the knapsack problem with cardinality constraints.
The FPTAS of \cite{CKPP} clearly has polynomial time in the size
of its input, and $\frac 1{\eps}$. Since the number of items that
we give to this algorithm as input is at most $n$, we can use this
FPTAS and still let our scheme have polynomial running time.

Since the approximated separation oracle that we described above
runs in polynomial time (polynomial in $n$ and $1\over \eps$) we
conclude that the approximated solution of the (primal) linear
program $x^*$ is obtained in polynomial time (again polynomial in
$n$ and $1\over \eps$). Since $x^*$ is a solution of a linear
program with an exponential number of variables, $x^*$ is given in
a compact representation, which is a list of non-zero components
of the solution, together with their values.
 The set of items in $I \setminus L_1$ is
packed according to the integral solution $y$ as described above.
It remains to pack $L_1$.  To do so, we pack each item of $L_1$
in a separate (dedicated) bin.  Note that there are $|L_1|=\lceil
n\eps^3 \rceil \leq n\eps^3+1$ such bins, and since $\opt \geq
{n\over k} \geq n\eps^2$ we conclude that the number of
additional bins (used to pack $L_1$) is at most $\eps\opt+1$, so
$\apx \leq \sum\limits_{C\in {\cal C}} y_C +|L_1|\leq
\sum\limits_{C\in {\cal C}} y_C +\eps\opt+1$. Therefore, it
suffices to bound the cost, implied by $y$, in terms of $\opt$.

Instead of using an optimal solution to the linear program (whose
value is a lower bound on $\opt(I')$, since $\opt(I')$ is a valid
solution to the linear program), we use a $1+\eps$-approximated
solution, and this degrades the value of the returned solution
within a factor of $1+\eps$ (i.e., $\sum\limits_{C\in\C}x_C^*\leq
(1+\eps)\cdot \opt(I')$).

We next bound $\sum\limits_{C\in {\cal C}} (y_C-x^*_C)$.  Note
that in the primal linear program there are at most $1\over
\eps^3$ constraints (not including non-negativity constraints),
and hence in a basic solution (a property that we can always
assume that $x^*$ satisfies) there are at most $1\over \eps^3$
positive components, and hence there are at most $1\over \eps^3$
fractional components. Therefore, $\sum\limits_{C\in {\cal C}}
(y_C-x^*_C) \leq {1\over \eps^3}$.

Therefore, $\apx \leq \sum\limits_{C\in {\cal C}} y_C +\eps\opt+1
= \sum\limits_{C\in {\cal C}} x^*_C +\sum\limits_{C\in {\cal C}}
(y_C-x^*_C) + \eps\opt +1 \leq (1+\eps)\cdot \opt(I') + {1\over
\eps^3} + \eps\opt +1 \leq (1+2\eps) \opt + {1\over \eps^3}+1$.
Hence, we conclude the following theorem.
\begin{theorem}
If $k\leq {1\over \eps^2}$, the above scheme is an AFPTAS for \pr.
\end{theorem}
\section{An AFPTAS for \bpr}
\label{bpr} In this section we use the similarity between the
APTAS of Caprara et al. \cite{CKP03} for \pr\ and the APTAS of
Epstein \cite{Eps06} for \bpr\ to develop our methods further. We
obtain an AFPTAS for \bpr\ using adaptations to the methods of the
previous section.

Without loss of generality we assume that $r_i \leq 1$ for all
$i$.  This is so as if there is an item with higher rejection
penalty, then it is better to pack this item in a separate
additional bin instead of rejecting it, and this situation is
true  for an item with a unit rejection penalty as well.
Therefore, by changing the rejection penalty of such an item to
1, we do not change the optimal solution or a (reasonable)
approximate solution.

Let $0< \eps \leq \frac 13$ be such that $1\over \eps$ is an
integer. An item $j$ is {\it large} if both $s_j\geq \eps$ and
$r_j \geq \eps$. All other items are {\it small}.  We denote by
$L$ the set of large items, and by $S$ the set of small items.

We perform rounding of the rejection penalties and the sizes of
the large items (only).  For $i=0,1,\ldots,\Delta={1\over
\eps^2}-{1\over \eps}$, and every large item $j\in L$, such that
$r_j \in [\eps+i\eps^2, \eps+(i+1)\eps^2)$, we round down the
rejection penalty $r_j$ to $\eps+i\eps^2$.  For a large item $j$,
denote the rounded rejection penalty of $j$ by $r'_j$, and for a
small item $j$ let $r'_j=r_j$. Define $I'$ to be the adapted
input, $\alg (I')$ be the cost of an arbitrary algorithm $\alg$ on
the input $I'$, and let $\alg '(I')$ be the cost of the same
algorithm on the original items. Then, Epstein \cite{Eps06}
showed the following:

\begin{lemma}[Lemma 1 in \cite{Eps06}] $\alg'(I')\leq
(1+\eps)\alg(I')$ and $\opt(I') \leq \opt(I)$.
\end{lemma}


For $i=0,1,\ldots ,\Delta$, let $L^i=\{j_1,\ldots,j_{n_i}\}$ be
the set of  large items with rounded rejection penalty
$\eps+i\eps^2$, such that $s_{j_1}\geq s_{j_2}\geq \cdots \geq
s_{j_{n_i}}$.  For each set $L^i$ such that $|L^i| \geq
\frac{1}{\eps^3}$, we perform linear grouping  separately.  That
is, for values of $i$ such that $|L^i| \geq \frac{1}{\eps^3}$, we
let $m={1\over \eps^3}$ and we partition $L^i$ into $m$ classes
$L^i_1,\ldots ,L^i_m$ such that $\lceil n_i\eps^3\rceil
=|L^i_1|\geq |L^i_2|\geq \cdots \geq |L^i_m|=\lfloor n_i\eps^3
\rfloor$, and $L^i_p$ receives the largest items from $L^i
\setminus \left[ L^i_1\cup \cdots \cup L^i_{p-1}\right]$.  For
every $i=0,1,\ldots ,\Delta$ and $j=2,3,\ldots ,n_i$ we round up
the size of the elements of $L^i_j$ to the largest size of any
element of $L^i_j$. For item $j$ of a set $L^i_p$ ($p \geq 2$),
we denote by $s'_{j}$ the rounded-up size of the item, which is
defined to be equal to the maximum size of any item in $L^i_p$.
For items in $L^i_1$ (for all $i$) we do not round the sizes, and
we denote $s'_j=s_j$ for all $j\in L^i_1$. For values of $i$ such
that $|L^i| < {1\over \eps^3}$, each large item of $L^i$ has its
own set $L^i_j$ such that $L^i_1$ is an empty set, and for a large
item $j\in L^i$ we let $s'_j=s_j$. We note that in both cases
(i.e., for large and small cardinalities of $L^i$) we have
$|L^i_1| \leq 2\eps^2 |L^i|$. For $j\in S$ we also denote
$s'_j=s_j$. We denote by $L_1=\cup_{i=0}^{\Delta} L^i_1$ and
$L'=L\setminus L_1$. By the above, we have $|L_1|\leq
2\eps^3|L|$. We consider the instance $I''$ consisting of the
items in $L'\cup S$ with the (rounded-up) sizes function $s'$ and
the rounded rejection penalty function $r'$. The items in $L_1$
are packed each in a separate bin. We have $\opt(I'')\leq
\opt(I')$, similarly to the previous sections. We next describe
the packing of the items in $I''$.

Given the instance $I''$, we let a configuration of a bin $C$ be
a (possibly empty) set of items of $L'$ whose total (rounded-up)
size is at most 1. We denote  the set of all configurations by
$\cal C$. Let $H=\{(\sigma_1,\rho_1),\ldots (\sigma_t,\rho_t)\}$
be  the set of different types of large items where $\sigma_j$
denotes the (rounded-up) size of an item with type
$(\sigma_j,\rho_j)$, and $\rho_j$ is its (rounded) rejection
penalty.  We have $|H|\leq \frac{1}{\eps^3}\cdot \Delta\leq
\frac{1}{\eps^5}$. For each $v\in H$ we denote by $n(v,C)$ the
number of items with type $v$ in $C$, and we denote by $n(v)$ the
number of items in $I''$ with type $v$. For a large item $j$, we
denote by $type(j)$ the type of $j$.

We denote the minimum  size of an item by $s_{min}= \min_{i\in S}
s'_i$ (recall that in \bpr, it is assumed that sizes of items are
strictly larger than 0), and as in the previous section, we let
$s'_{min}=\max\{\frac{1}{(1+\eps)^t}|t\in \mathbb{Z}, \
\frac{1}{(1+\eps)^t}\leq s_{min}\}$. The value
$\log_{1+\eps}{\frac{1}{s'_{min}}}$ is polynomial in the size of
the input. We define the following set
$\W=\{\frac{1}{(1+\eps)^t}|0\leq t \leq
\log_{1+\eps}{\frac{1}{s'_{min}}}+1\}$. A {\it window} is defined
as a member of $\W$.  $\W$ is also called the set of all possible
windows. Then, $|\W| \leq \cdot (\log_{1+\eps} {1\over
s'_{min}}+2) $. Since windows are scalars, they can be compared
with respect to their size.

Note that each bin packed with large items  according to a
configuration $C$ typically leaves space for small items.  For a
configuration $C$ we denote the {\it main window of $C$} to be
$w(C)$, which can be seen as an approximation of the available
size for small items in a bin with configuration $C$. We define
it as follows. Assume that the total (rounded-up) size of the
items in $C$ is $s'(C)$.  Then, $w(C)={1\over (1+\eps)^t}$ where
$t$ is the maximum integer such that  $0 \leq t \leq
\log_{1+\eps}{\frac{1}{s'_{min}}}+1$ and that $s'(C) + {1\over
(1+\eps)^t} \geq 1$. The main window of a configuration is a
window (i.e., belongs to $\W$), but $\W$ may include windows that
are not the main window of any configuration. We note that $|{\cal
W}|$ is polynomial in the input size and in $1\over \eps$,
whereas $|{\cal C}|$ may be exponential in $1\over \eps$,
specifically, $|{\cal C}| \leq ({1\over \eps^5}+1)^{1/\eps}$,
since in configuration there are up to $1\over \eps$ large items
of $|H| \leq{ 1\over \eps^5}$ types. We denote the set of windows
that are actual main windows of at least one configuration by
${\cal {W}'}$. Similarly to the previous section, we define a
linear program that allows the usage of any window in $\W$ and
later modify the linear program and the solution to this linear
program (that we obtain) to use only windows of ${\cal W}'$.

We define a generalized configuration $\tilde{C}$ as a pair
$\tilde{C}=(C,w=\frac 1{(1+\eps)^t})$, for some configuration
$C$, and some $w \in \W$. The generalized configuration
$\tilde{C}$ is valid if $w\leq w(C)$. The set of all valid
generalized configurations is denoted by $\tilde{\C}$. For every
$W\in {\cal W}$ denote by $C(W)$ the set of generalized
configurations such that $W$ is their window, i.e., $C(W)=\{
\tilde{C}=(C,w)\in {\tilde{\C}} : w=W\}$.

We next consider the following linear program.   In this linear
program we have a variable $x_{\tilde{C}}$ denoting the number of
bins with the generalized configuration $\tilde{C}$,  variables
$Y_{i,W}$ indicating if the small item $i$ is packed in a window
of type $W$, and variables $z_i$ indicating that item $i$ is
rejected (for both small and large items).

\begin{eqnarray}
\min & \sum\limits_{{\tilde{C}}\in {{\tilde{\C}}}} x_{\tilde{C}}+\sum\limits_{i\in L'\cup S} r'_iz_i& \nonumber \\
s.t.& \sum\limits_{{\tilde{C}=(C,w)}\in {{\tilde{\C}}}}n(v,C) x_{\tilde{C}} +\sum\limits_{j: type(j)=v} z_j  \geq n(v) & \forall v\in H\label{c1}\\
& \sum\limits_{W\in {\cal W}} Y_{i,W} +z_i \geq 1& \forall i\in S\label{c2}\\
& W\cdot \sum\limits_{{\tilde{C}}\in C(W)} x_{\tilde{C}}  \geq \sum\limits_{i\in S} s'_i\cdot Y_{i,W} & \forall W\in {\cal W}\label{c3}\\
& x_{\tilde{C}} \geq 0& \forall {{\tilde{C}}\in {{\tilde{\C}}}} \nonumber \\
& Y_{i,W} \geq 0&  \forall W\in {\cal W}, \forall i \in
I''\nonumber \\
& z_i \geq 0& \forall i \in I''.\nonumber
\end{eqnarray}
Constraints (\ref{c1}) and (\ref{c2}) ensure that each item (large
or small) of $I''$ is packed or rejected by the solution.
Constraints (\ref{c3}) ensure that the total size of the small
items that we decide to pack in a window of type $W$ is not larger
than the total available space allocated to the windows of small
items, that is, the number of bins that are packed according to a
generalized configuration that has this window is large enough.
We note that $\opt(I'')$ implies a feasible solution to the above
linear program that has the cost $\opt(I'')$, since the packing
of the small items (including the specification of the subset of
rejected items) clearly satisfies the constraints (\ref{c3}), and
the packing of large items (including the specification of the
subset of rejected items) satisfies the constraints (\ref{c1}).
Moreover, it implies a solution to the linear program in which
all variables $x_{\tilde{C}}$, that correspond to generalized
configurations $\tilde{C}=(C,w)$, for which $w$ is not the main
window of $C$, are equal to zero, and all variables $Y_{i,w}$
where $w \notin \W'$ are equal to zero as well.

Once again (similarly to the AFPTAS for \pr) we invoke the column
generation technique of Karmarkar and Karp \cite{KK82} as follows.
The above linear program has exponential number of variables and
polynomial number of constraints (neglecting the non-negativity
constraints). Instead of solving the linear program we solve its
dual program (that has a polynomial number of variables and an
exponential number of constraints). The variables $\alpha_v$
correspond to the item types in $H$, their intuitive meaning can
be seen as weights of these items.  The variables $\beta_i$
correspond to the small items, and their intuitive meaning can be
seen as weights of these items.  For each $W\in {\cal W}$ we have
a dual variable $\gamma_W$.  Using these dual variables, the dual
linear program is as follows.

\begin{eqnarray}
\max & \sum\limits_{v \in H} n(v) \alpha_v+ \sum\limits_{i\in S} \beta_i&\nonumber \\
s.t. & \sum\limits_{v\in H} n(v,C) \alpha_v + w\cdot \gamma_w \leq
1& \forall {\tilde{C}=(C,w)} \in {\tilde{\C}}\label{d1}\\
& \beta_i - s'_i\gamma_W \leq 0&\forall i\in S, \ \forall W\in {\cal W}\label{d2}\\
& \beta_i \leq r'_i& \forall i\in S\label{d3}\\
&\alpha_v \leq r'_i& \forall v\in H,\ \forall i\in L': type(i)=v
\label{d4}\\
 & \alpha_v \geq 0& \forall v\in H\nonumber
\\
& \beta_i \geq 0& \forall i\in S\nonumber\\
 &\gamma_W \geq 0&
\forall W\in {\cal W}.\nonumber
\end{eqnarray}
First note that there is a polynomial number of constraints of
type (\ref{d2}), (\ref{d3}) and (\ref{d4}), and therefore we
clearly have a polynomial time separation oracle for these
constraints. If we would like to solve the above dual linear
program (exactly) then using the ellipsoid method we need to
establish the existence of a polynomial time separation oracle for
the constraints (\ref{d1}).  However, we are willing to settle on
an approximated solution to this dual program. To be able to apply
the ellipsoid algorithm, in order to solve the above dual problem
within a factor of $1+\eps$, it suffices to show that there exists
a polynomial time algorithm (polynomial in $n$, $1\over \eps$ and
$\log\frac 1{ s'_{min}}$) such that for a given solution
$a^*=(\alpha^*,\beta^*,\gamma^*)$ decides whether $a^*$ is a
feasible dual solution (approximately).

That is, it either provides a generalized configuration
$\tilde{C}=(C,w) \in \tilde{\C}$ for which $\sum\limits_{v\in H}
n(v,C) \alpha^*_v + w\gamma^*_w> 1$, or outputs that an
approximate infeasibility evidence does not exist, that is, for
all generalized configurations $\tilde{C}=(C,w) \in \tilde{\C}$,
$\sum\limits_{v\in H} n(v,C) \alpha^*_v + w\gamma^*_w \leq 1+\eps$
holds. In such a case, $a^*\over 1+\eps$ is a feasible dual
solution that can be used.

Such a configuration $\tilde{C}$ can be found by the following
procedure:  For each $W \in {\cal W}$ we look for a configuration
$C \in \C$ such that $(C,W)$ is a valid generalized configuration,
and $\sum\limits_{v\in H} n(v,C) \alpha^*_v$ is maximized. If a
configuration $C$ that is indeed found, the generalized
configuration, whose constraint is checked, is $(C,W)$.

To find $C$, we invoke an FPTAS for the standard knapsack problem
with the following input: The set of items is $H$ where for each
$v\in H$ there is a volume $\alpha^*_v$ and a size $v$, the goal
is to pack a multiset of the items (an item can appear at most a
given number of times), so that the total volume is maximized,
under the condition that the total (rounded-up) size of the
multiset should be smaller than $1-\frac{W}{1+\eps}$, unless
$W<s'_{min}$, where the total size should be at most 1 (in this
case, the window leaves space only for items of size zero). Since
the number of applications of the FPTAS for the knapsack problem
is polynomial (i.e., $|{\cal W}|$), this algorithm runs in
polynomial time.

 If it finds a solution,
that is, a configuration $C$, with total volume greater than
$1-W\gamma^*_W$, we argue that $(C,W)$ is indeed a valid
generalized configuration, and this implies that there exists a
generalized configuration, whose dual constraint (\ref{d1}) is
violated.

By the definition of windows, the property $W<s'_{min}$ is
equivalent to $W=\frac{s'_{min}}{1+\eps}$, which is the smallest
size of window (which forms a valid generalized configuration with
any configuration).  If $W\geq s'_{min}$, recall that the main
window of $C$, $w(C)$ is chosen so that
$s'(C)+{w(C)} \geq 1$, and that $C$ is chosen by the algorithm for
the knapsack problem, so that $s'(C)<1-\frac{W}{1+\eps}$. We get
$1-w(C) \leq s'(C)<1-\frac{W}{1+\eps}$ and therefore $W <
(1+\eps)w(C)$, i.e.,  $W \leq w(C)$ (since the sizes of windows
are integer powers of $1+\eps$), so we conclude that $(C,W)$ is  a
valid generalized configuration.
Thus in this case we found that this solution is a configuration
whose constraint in the dual linear program is not satisfied, and
we can continue with the application of the ellipsoid algorithm.

Otherwise, for any window $W$, any configuration $C$ of total
rounded-up size less than $1-\frac{W}{1+\eps}$ (or at most 1, if
$W<s'_{min}$), has a volume of at most
$(1+\eps)(1-W\gamma^*_W)\leq (1+\eps)-W\gamma^*_W$. We prove that
in this case, all the constraints of the dual linear program are
satisfied by the solution $a^*\over 1+\eps$. Consider a valid
generalized configuration $\tilde{C}=(C,\tilde{w})$. We have
$\tilde{w}\leq w(C)$, where $w(C)$ is the main window of $C$. If
$w(C)<s'_{min}$, then $\tilde{w}=w(C)$. Since $s'(C)\leq 1$ for
any configuration, $C$ is a possible configuration to be used with
the window $\tilde{w}$ in the application of the FPTAS for
knapsack. Assume next that $\tilde{w} <1$, then when the FPTAS for
knapsack is applied on $W$, $C$ is a configuration that is taken
into account for $W$ since $s'(C)<1-\frac{w(C)}{1+\eps}\leq
1-\frac{\tilde{w}}{1+\eps}$, where the first inequality holds by
definition of $w(C)$. If $\tilde{w} =1$ then $1 \geq w(C) \geq
\tilde{w} =1$, so $w(C)=1$. A configuration $C_1$ that contains at
least one large item satisfies $s'(C_1)\geq \eps$, so
$s'(C_1)+\frac{1}{1+\eps} \geq \frac{1+\eps+\eps^2}{1+\eps}>1$.
Therefore if the main window of a configuration is of size 1, this
configuration is empty. We therefore have that $C$ is an empty
configuration, thus $s'(C)=0$. This empty configuration is
considered with any window.

We denote by $(X^*,Y^*,Z^*)$ the solution to the primal linear
program that we obtained.  Since its cost is a $(1+\eps)$
approximation for the optimal solution to the linear program, we
conclude that $\sum\limits_{{\tilde{C}\in
{\tilde{\C}}}}X^*_{\tilde{C}} + \sum\limits_{i\in L'\cup S}
r'_iZ^*_i \leq (1+\eps)\opt(I'')$.

We modify the solution to the primal linear program, into a
different feasible solution of the linear program, without
increasing the objective function. We create a list of generalized
configurations whose $X^*$ component is positive. From this list
of generalized configurations, we find a list of windows that are
the main window of at least one  configuration induced by a
generalized configuration in the list. This list of windows is a
subset of $\cal{W}'$ defined above. We would like the solution to
use only windows from $\cal{W}'$.

We modify the $X^*$ and $Y^*$ components, while the $Z^*$
components are not modified. The new solution will have the
property that any non-zero components of $X^*$, $X^*_{\tilde{C}}$
corresponds to a generalized configuration $\tilde{C}=(C,w)$, such
that $w\in \cal{W}'$. We still allow generalized configurations
$\tilde{C}=(C,w)$ where $w$ is not the main window of $C$, as long
as $w\in \W'$. This is done in the following way. Given a window
$w' \notin \W'$, we define $B_{w'}=\sum\limits_{\tilde{C''}\in
C(w')}X^*_{\tilde{C''}}$. The following is done in parallel for
every generalized configuration $\tilde{C'}=(C,w')$, where
$w'\notin \W'$ and such that $X^*_{\tilde{C'}}>0$, where the main
window of $C$ is $w \geq w'$ (but $w'\neq w$). We let
$\tilde{C}=(C,w)$. The windows allocated for small items need to
be modified first, thus an amount of
$\frac{X^*_{\tilde{C'}}}{B_{w'}}Y_{i,w'}$ is transferred from
$Y_{i,w'}$ to $Y_{i,w}$. We modify the values $X^*_{\tilde{C'}}$
and $X^*_{\tilde{C}}$ as follows. We increase the value of
$X^*_{\tilde{C}}$ by an additive factor of $X^*_{\tilde{C'}}$ and
let $X^*_{\tilde{C'}}=0$.

To show that the new vector $(X^*,Y^*,Z^*)$ still gives a feasible
solution of the same value of objective function, we consider the
modifications. The sum of components of $X^*$ does not change at
all in the above process, thus the value of the objective function
is the same. Moreover, for every configuration $C$, the sum of
components $X^*$, that correspond to  generalized configurations
whose configuration of large items is $C$, does not change. Thus
the constraints (\ref{c1}) still hold. We next consider the
constraint (\ref{c2}) for $i$, for a given small item $i\in S$.
Since the sum of variables $Y^*_{i,W}$ does not change, this
constraint still holds.

As for constraints (\ref{c3}), for a window $w \notin \cal{W'}$,
the right hand side of each such constraint became zero. On the
other hand, for windows in $\cal{W'}$, every increase in some
variable $X^*_{\tilde{C}}$ for $\tilde{C}=(C,w)$, that is
originated in a decrease of $X^*_{\tilde{C'}}$ for
$\tilde{C'}=(C,w')$ is accompanied with an increase of
$\frac{X^*_{\tilde{C'}}}{\sum\limits_{\tilde{C''}\in
C(w')}X^*_{\tilde{C''}}}Y^*_{i,w'}=\frac{X^*_{\tilde{C'}}}{B_{w'}}Y^*_{i,w'}
$ in $Y^*_{i,w}$, for every $i\in S$, thus is, an increase of
$\sum\limits_{i\in S}\frac{x^*_{\tilde{C'}}}{B_{w'}} s'_i\cdot
Y^*_{i,w'}$ in the right hand size of the constraint (\ref{c3})
for $w$, and an increase of $w \cdot x^*_{\tilde{C'}}$ in the left
hand side. Since we have $w \cdot B_{w'} \geq w' \cdot B_{w'} \geq
\sum\limits_{i\in S} s'_i\cdot Y^*_{i,w'}$ before the modification
occurs (since constraint (\ref{c3}) holds for the solution before
modification for the window $w'$), we get that the increase of the
left hand side is no smaller than the increase in the right hand
side.

Now, we can delete the constraints of (\ref{c3}) that correspond
to windows in ${\cal W} \setminus {\cal W}'$.  In the resulting
linear program we consider a basic solution that is not worse than
the solution we obtained above. Such  a basic solution can be
found in polynomial time.  We denote this basic solution by
$(\X^*,\Y^*,\Z^*)$.

We apply several steps of rounding to obtain a feasible packing
of the items into bins.  We first round up $\X^*$.  That is,
denote by $\hat{X}$ the vector such that
$\hat{X}_{\tilde{C}}=\lceil \X^*_{\tilde{C}} \rceil$ for all
$\tilde{C}\in {\tilde{\cal C}}$.  Moreover, each small item $i\in
S$ such that $(\Y^*_{i,W})_{W\in {\cal W}}$ is fractional, is
packed using a dedicated bin. We modify the value of
$\hat{x}_{\tilde{C}}$ for $\tilde{C}$ that corresponds to an
empty configuration $C$, together with the window $1$, to reflect
the additional bins that accommodate the small items that were
previously packed fractionally. We modify the values
$(\Y^*_{i,W})_{W\in {\cal W}}$ so that every item $i$ which is
packed into a new bin has $\hat{Y}_{i,W}=0$ for all $W$ except
for $W=1$, for which $\hat{Y}_{i,W}=1$. For all other variables
$Y_{i,W}$ we define $\hat{Y}_{i,W}=\Y^*_{i,W}$. For every $i\in
I'$ we let $\hat{Z_i}=\Z^*_i$.

We next bound the increase in the cost due to this rounding.

\begin{lemma}\label{ll1}
$\sum\limits_{C\in {\cal C}} \hat{X}_C\leq \sum\limits_{C\in {\cal
C}}\X^*_C + |H|+|{\cal W'}|$.
\end{lemma}
\begin{proof}
Consider now the primal linear program, where constraints
(\ref{c3}) exist only for windows in $\W'$, the variables
$x_{\tilde{C}}$ exist only for generalized configurations
$\tilde{C}=(C,w)$ where $w \in \W'$, and the variables $Y_{i,W}$
exists only for $W\in\W'$. The basic solution $(\X^*,\Y^*,\Z^*)$
is a feasible solution for this linear program. In the primal
linear program there are $|H|+|{\cal W'}|+|S|$ inequality
constraints, and hence in a basic solution there are at most
$|H|+|{\cal W'}|+|S|$ basic variables. For every $i\in S$, there
is at least one variable associated with $i$ that is a basic
variable. This variable is either $Z_i$ or $Y_{i,W'}$ for some
$W'$, and therefore the number of basic variables from the $x$
components and additional basic variables from the $Y$ and $Z$
components is at most $|H|+|{\cal W'}|$. Hence the sum of the
number of fractional components among the $x_{\tilde{C}}$
variables, and the number of small items such that the vector
$(\Y^*_{i,W})_W$ is non-integral is at most $|H|+|{\cal W'}|$.
This is an upper bound on the difference in the objective values
of the two solutions and the claim follows.
\end{proof}


Our scheme returns a solution that packs $\hat{X}_{\tilde{C}}$
bins with configuration $\tilde{C}$. Each large item of the
rounded-up instance is replaced by the corresponding item of
$I$.  We clearly use at most $\sum\limits_{\tilde{C}\in
{\tilde{\C}}} \hat{X}_{\tilde{C}}$ bins in this way. We next pack
each item of $L_1$ by its own bin (if $L_1$ is non-empty). We
denote the resulting solution by $SOL_{large}$.

\begin{lemma}The cost of $SOL_{large}$ is at most
$\sum\limits_{\tilde{C}\in {\tilde \C}}\hat{X}_{\tilde{C}}
+\sum\limits_{i\in I'} r'_i \hat{Z}_i+ \eps\opt$.
\end{lemma}
\begin{proof}
It suffices to show that $|L_1| \leq \eps\opt$.  To see this last
claim note that $|L_1|\leq 2|L|\eps^3$ and each item in $L$ has
both a size of at least $\eps$ and a rejection cost of at least
$\eps$ and therefore $\opt \geq |L|\eps$, and therefore $|L_1|
\leq 2\eps^2 \opt$ and the claim follows since $\eps <{1/2}$.
\end{proof}

\begin{corollary}
The cost of $SOL_{large}$  is at most $(1+2\eps)\opt + |H|+|{\cal
W'}|$.
\end{corollary}

We next consider the packing of the small items that are supposed
to be packed (according to $\hat{Y}$) in bins with window $W$.
Assume that there are $X(W)$ such bins (i.e.,
$X(W)=\sum\limits_{\hat{C}:\tilde{C}=(C,W)}
\tilde{X}_{\tilde{C}}$). Denote by $S(W)$ the set of small items
that we decided to pack in bins with window $W$ (for some of
these items we will change this decision in the sequel). Then, by
the feasibility of the linear program we conclude that
 $\sum\limits_{i\in S(W)} s'_i \leq W \cdot X(W)$.
We allocate almost all $S(W)$ to the $X(W)$ bins with window $W$
such that the total size of the items in each such bin is at most
$1+\frac{\eps W}{1+\eps}$, exactly as in the previous section.
As in the previous section, we create an intermediate solution
$SOL_{inter}$ by removing the largest small item from each such
bin. Each removed item is small and therefore either its rejection
penalty is at most $\eps$ or its size is at most $\eps$. We pack
the removed small items with size at most $\eps$ in new bins,
packing $1\over \eps$ such items in a bin, except perhaps the last
such bin, and the other removed items are rejected (incurring a
rejection penalty of at most $\eps$ for each such item).
  The
total cost of this intermediate solution is therefore at most
$(1+\eps)\cdot \left( (1+2\eps)\opt + |H|+|{\cal W'}| \right) +1$
(the last bin can contain less than $1\over \eps$ such removed
items).  As in the previous section, after the largest small items
is removed from each bin, the total size of small items assigned
to such bin is at most $W$.

The intermediate solution is infeasible because our definition of
$W$ is larger than the available space for small items in such
bin. We create the final solution $SOL_{final}$ using the same
process as in the previous section. That is, given a bin such that
the intermediate solution packs to it large items according to
configuration $C$, and small items with total size at most $w(C)$,
we remove the small items and pack them back until the first item
that causes an excess. The first items whose rejection penalties
are smaller than $\eps$ are rejected. The other ones are packed in
separate bins (each such separate bin will contain $1\over \eps$
such first items for different bins of $SOL_{inter}$). Similarly
to the above argument these are feasible bins and they add an
additive factor of $\eps$ times the cost of $SOL_{inter}$ to the
total cost of the packing (plus 1). The remaining unpacked items
in a bin of $SOL_{inter}$ have total size of at most $\frac{\eps
W}{1+\eps} \leq \eps$. Therefore, we can pack the unpacked items
of $1\over \eps$ bins of $SOL_{inter}$ using one additional bin.
In this way we get our final solution $SOL_{final}$.  We note
that the cost of $SOL_{final}$ is at most $(1+2\eps)$ times the
cost of $SOL_{inter}$ plus one.  Therefore the cost of
$SOL_{final}$ is at most $ (1+2\eps) \left( (1+\eps)\cdot \left(
(1+2\eps)\opt + |H|+|{\cal W}| \right) +1 \right) +1 \leq
(1+10\eps)\opt + 4(|H|+|{\cal W}|+1) \leq (1+10\eps)\opt +
4{1\over \eps^5}+ 4({1\over \eps^5}+1)^{1/\eps} +1$ where the
last inequality holds by $\eps < {1\over 3}$, $|H|\leq {1\over
\eps^5}$ and $|{\cal W}|\leq |{\cal C}| \leq ({1\over
\eps^5}+1)^{1/\eps}$  Therefore, we have established the
following theorem.

\begin{theorem}
The above scheme is an AFPTAS for  \bpr.
\end{theorem}

\end{document}